\documentclass{article}
\usepackage{graphicx}
\usepackage{authblk}
\usepackage{bm}
\usepackage{verbatim}
\usepackage{algorithmicx}
\usepackage{algorithm}
\usepackage{subfigure}
\usepackage{bm}
\usepackage{dsfont}
\usepackage[utf8]{inputenc}
\usepackage{mathtools}
\usepackage{natbib}
\usepackage{amsfonts}
\usepackage{url}
\usepackage{multirow}
\usepackage{subfigure}
\usepackage{pdfpages}
\usepackage[parfill]{parskip}
\usepackage{tikz} 
\usepackage{abstract}
\usepackage[pagebackref]{hyperref} 
\usepackage{listings}
\usepackage{siunitx}
\usepackage{algorithm}
\usepackage{algpseudocode}

\numberwithin{equation}{section}

\newcommand{\norm}[1]{\left\lVert#1\right\rVert}
\newenvironment{proof}{\begin{trivlist} \item[]
{\bf Proof.}}{\nolinebreak
\hfill \rule{2mm}{2mm} \end{trivlist}}
\newtheorem{definition}{Definition}[subsection] 
\newtheorem{theorem}[definition]{Theorem}
\newenvironment{theorem*}[1]{{\bf Theorem #1} \begin{itshape}}{\end{itshape}}

\newenvironment{corollary*}[1]{{\bf Corollary #1} \begin{itshape}}{\end{itshape}}

\newenvironment{proposition*}[1]{{\bf Proposition #1} \begin{itshape}}{\end{itshape}}

\providecommand{\keywords}[1]
{
  \small
  \textbf{\textit{Keywords---}} #1
}

\title{High-Dimensional Changepoint Detection via a Geometrically Inspired Mapping \thanks{Grundy is grateful for the support of the Engineering and Physical Sciences Research Council (grant number EP/L015692/1). The authors also acknowledge Royal Mail Group Ltd for financial support, and are grateful to Jeremy Bradley in Royal Mail GBI Data Science for helpful discussions.}}
\date{}
\author[1]{Thomas Grundy}
\author[2]{Rebecca Killick}
\author[3]{Gueorgui Mihaylov}
\affil[1]{STOR-i Centre for Doctoral Training, Lancaster University, UK (t.grundy1@lancaster.ac.uk)}
\affil[2]{Mathematics and Statistics Department, Lancaster University, UK}
\affil[3]{Royal Mail GBI Data Science, UK}

\begin{document}
\maketitle

\begin{abstract}
High-dimensional changepoint analysis is a growing area of research and has applications in a wide range of fields. The aim is to accurately and efficiently detect changepoints in time series data when both the number of time points and dimensions grow large. Existing methods typically aggregate or project the data to a smaller number of dimensions; usually one. We present a high-dimensional changepoint detection method that takes inspiration from geometry to map a high-dimensional time series to two dimensions. We show theoretically and through simulation that if the input series is Gaussian then the mappings preserve the Gaussianity of the data. Applying univariate changepoint detection methods to both mapped series allows the detection of changepoints that correspond to changes in the mean and variance of the original time series. We demonstrate that this approach outperforms the current state-of-the-art multivariate changepoint methods in terms of accuracy of detected changepoints and computational efficiency. We conclude with applications from genetics and finance.
\end{abstract}
\section{Introduction}\label{sec:Introduction}
Time series data often have abrupt structural changes occurring at certain time points, known as changepoints. To appropriately analyze, model or forecast time series data that contain changes we need to be able to accurately detect where changepoints occur. High-dimensional changepoint analysis aims to accurately and efficiently detect the location of changepoints as both the number of dimensions and time points increase. High-dimensional changepoint analysis is an ever-growing research area and has multiple applications including finance and economics \citep{Modisett2010}; longitudinal studies \citep{Terrera2011} and genetics \citep{Bleakley2011}.

Changepoint analysis in the univariate setting is a well-studied area of research with early work by \cite{Page1954} and overviews can be found in \cite{Eckley2011} and \cite{Brodsky2013}. The multivariate extension has received less attention, see \cite{Truong2018} for a recent review. One major challenge with high-dimensional changepoint analysis is the computational burden of an increasing number of dimensions. To partially reduce this computational burden, a common assumption is that changepoints are assumed to occur in all series simultaneously \citep{Maboudou-Tchao2013}; a sparse set of series \citep{Wang2018}; or a dense set of series \citep{Zhang2010}. Within these settings, a common approach is to first project the time series to a single dimension and then use a univariate changepoint method on the projected time series. For example, \cite{Zhang2010}, \cite{Horvath2012} and \cite{Enikeeva2019} consider an $l_2$-aggregation of the CUSUM statistic while \cite{Jirak2015} considers an $l_\infty$-aggregation that works well for sparse changepoints. A recent advancement was the Inspect method proposed by \cite{Wang2018} who aim to find an optimal projection direction of the CUSUM statistic to maximize a change in mean. 

Current projection methods are generally limited to detecting changes in a single parameter, usually the mean. Therefore, these methods cannot be used in many practical scenarios where multiple features of the time series change. An alternative, nonparametric approach was taken in \textcolor{black}{\cite{Matteson2014}} where $U$-Statistics were used to segment the time series. As this is a non-parametric method it can detect different types of changes in distribution but becomes computationally infeasible as the number of time points increases. The methods above almost exclusively use a Binary Segmentation approach \citep{Scott1974,Vostrikova1981}, or derivations thereof \citep{Fryzlewicz2014}, to detect multiple changepoints. \textcolor{black}{This can lead to poor detection rates as conditional identification of changes can lead to missing or poor placement of changepoints due to factors such as masking. This occurs when a large change is masked by two smaller changes on either side acting in opposite directions; this idea is explained further in \cite{Fryzlewicz2014}.}

A key novelty in this paper is to map a given high-dimensional time series onto two dimensions instead of one. Inspired by a geometric representation of data, we map each high-dimensional time vector to its distance and angle from a fixed pre-defined reference vector based upon the standard scalar product. These mappings show shift and shape changes in the original data corresponding to mean and variance changes. Given the geometric inspiration, we denote the method GeomCP throughout.

In Section \ref{sec:Methodology}, we set up the high-dimensional changepoint problem before defining the geometric mappings used in GeomCP. Also, we discuss an alternative approach to Binary Segmentation that can be applied to the univariate mapped series. An extensive simulation study is performed in Section \ref{sec:Simulation}, which compares GeomCP to competing available multivariate changepoint methods, Inspect \citep{Wang2018} and E-Divisive \textcolor{black}{\citep{Matteson2014}}. Section \ref{sec:Applications} presents two applications from genetics and finance. Section \ref{sec:Conclusion} gives concluding remarks. 

\section{Methodology}\label{sec:Methodology}
In this section, we set up the high-dimensional changepoint problem for our scenario. We define our new method, GeomCP, and discuss how changes in high-dimensional time series manifest themselves in the mapped time series. We then suggest an appropriate univariate changepoint detection method for detecting changes in the mapped time series - although practically others could be used.

Before proceeding, we define some notation used throughout the paper. We define the $\mathds{1}_p$ vector as a $p$-dimensional vector where each entry is 1 and the number of dimensions, $p$, is inferred from context. For a vector, $\bm{y}=(y_1,\ldots,y_p)^T$, we define the $l_q$-norm as $\norm{\bm{y}}_q:=\left(\sum\limits_{j=1}^p|y_j|^q\right)^\frac{1}{q}$ for $q\in[1,\infty)$. We define $\left\langle\cdot,\cdot\right\rangle$ as the standard scalar product such that for vectors $\bm{x}$ and $\bm{y}$ we have $\left\langle\bm{x},\bm{y}\right\rangle=\sum\limits_{j=1}^px_jy_j$. Finally, the terms variables, series and dimensions shall be used interchangeably to indicate the multivariate nature of the problem.

\subsection{Problem Setup}\label{sec:ProblemSetup}
We study the time series model where $\bm{Y_1},\hdots,\bm{Y_n}$ are independent, $p$-dimensional time vectors that follow a multivariate Normal distribution where,
\[\bm{Y_i}\sim N_p(\boldsymbol\mu_i,\boldsymbol\sigma^2_i\boldsymbol I_p),\;\;\;\;\;1\leq i\leq n\;.\]
We assume there are an unknown number of changepoints, $m$, which occur at locations $\tau_{1:m}=\left(\tau_1,\ldots,\tau_m\right)$. These changepoints split the data into $m+1$ segments, indexed $k$, that contain piecewise constant mean and variance vectors, $\boldsymbol\mu_k$ and $\boldsymbol\sigma_k^2$. Note, we assume a diagonal covariance matrix so the covariance matrix can be described by the variance vector and the identity matrix. We define $\tau_0=0$ and $\tau_{m+1}=n$ and assume the changepoints are ordered so, $\tau_0=0<\tau_1<\ldots<\tau_m<\tau_{m+1}=n$.

The following section introduces the geometric intuition and mappings used within GeomCP. These mappings reduce the dimension of the problem to make the problem computationally feasible as $n$ and $p$ grow large. 

\subsection{Geometric Mapping}\label{sec:GeometricMapping}
When analyzing multivariate time series from a geometric viewpoint, we seek to exploit relevant geometric structures defined in the multi-dimensional space. Here we aim to detect changepoints in the mean and variance vectors of multivariate Normal random variables; therefore, we wish to utilize geometric properties that capture these changes.

\textcolor{black}{A change in the mean vector of our data generating process will cause a location shift of the data points in the multi-dimensional space. Consider a distance between each data point and some fixed reference point, if the data points are shifted in the multi-dimensional space then their distance to the reference point would be expected to change. Hence, we can detect when the mean vector of the data generating process changes by observing a change in the distances. For a change in distance not to occur after an underlying mean change, the new mean vector must remain exactly on the same $(p-1)$-sphere (centered in the reference point) that the old mean vector lay on. Given that the computation of the mean vector is a linear operator on the multivariate time series, the requirement to lie on the same sphere (a quadric in $\mathbb{R}^p$) is highly non-generic from a geometric prospective. As a result, these scenarios are rare especially in high-dimensions.}

\textcolor{black}{A change in the covariance of our data generating process will cause a change in the shape of the data points. More specifically in our setup, a change in the variance would cause the shape of the data points to expand or contract. Consider the angle between each data vector and a reference vector, as the shape of the data points expands (contracts) the angles will become more (less) varied. Hence, we can detect changes in the variance of the data generating process by detecting changes in the angles.} 
  
By using distances and angles, we can map a $p$-dimensional time series to two dimensions. To calculate these mappings, we need a pre-specified reference vector to calculate a distance and angle from. Naturally, one may think to use the mean of the data points. However, this requires a rolling window to estimate the mean of data points prior to the point being mapped. Not only does this introduce tuning parameters, such as the size of the rolling window, but will result in spikes in the distance and angle measures at changepoints. To detect changepoints, we would need a threshold for these spikes and calculating such a threshold is a non-trivial task, hence, we seek an alternative. 

\textcolor{black}{We propose setting the reference vector to be a fixed vector, $\bm{y_0}$. We then translate all the points based upon this fixed reference vector,
\begin{equation}\label{eq:PointRelocation}
  y'_{i,j}=y_{i,j}-(\min\limits_iy_{i,j}-y_{0,j})\;,\;i\in[1,\ldots,n]\;,\;j\in[1,\ldots,p]\;.
\end{equation}
This results in a data-driven reference vector. We choose to set $\bm{y_0}=\mathds{1}$ as this bounds the angle measure between 0 and $\pi/4$ meaning we do not get vectors close to the origin facing in opposite directions causing non-standard behavior within a segment. Moreover, having a non-zero element in every entry of $\bm{y_0}$ ensures changes in the individual series will manifest in the angle measure. Note due to the translation in \eqref{eq:PointRelocation}, the choice of $\bm{y_0}$ does not affect the distance measure. Throughout we assume the reference vector is set as $\bm{y_0}=\mathds{1}$.} 

For data points in the same segment, we would expect their distances and angles to the reference vector to have the same distribution. When a mean (variance) change occurs in the data, this leads to a shift (spread) in the data, hence, the distances (angles) will change. Therefore, by detecting changes in the distances and angles, using an appropriate univariate changepoint method, we recover changepoints in the $p$-dimensional series.

We define our distance and angle measures based upon the standard scalar product. To obtain our distance measure, $d_i$, we perform a mapping, $\delta:\mathbb{R}^p\rightarrow\mathbb{R}^1_{>0}$,\textcolor{black}{
\begin{equation}\label{eq:Distance}
  d_i=\delta(\bm{y_i})=\sqrt{\left\langle(\bm{y'_i}-\mathds{1}),(\bm{y'_i}-\mathds{1})\right\rangle}\;,
\end{equation}
}
which is equivalent to $\norm{\bm{y'_i}-\mathds{1}_p}_2$. 

To obtain our angle measure, $a_i$, we perform a mapping $\alpha:\mathbb{R}^p\rightarrow[0,\frac{\pi}{4}]$,\textcolor{black}{
\begin{equation}\label{eq:Angle}
  a_i=\alpha(\bm{y_i})=\cos^{-1}\left(\frac{\langle\bm{y_i}^{\prime},\mathds{1}\rangle}{\sqrt{\langle\bm{y_i}^{\prime},\bm{y'_i}\rangle}\sqrt{\langle\mathds{1},\mathds{1}\rangle}}\right)\;,
\end{equation}
}
which is the principal angle between $\bm{y'_i}$ and $\mathds{1}$.

By using the standard scalar product we are incorporating information from each series in the distance and angle measures. As such, we would expect GeomCP to perform well in scenarios where a dense set of the series change at each changepoint. This idea will be explored further and verified in Section \ref{sec:Simulation}.

\subsection{Analyzing Mapped Time Series}\label{sec:AnalyzingMappedTS}
Understanding the distributional form of the distance and angle mappings will aid in the choice of univariate changepoint methods. Under our problem setup, Theorem \ref{theorem:distance} shows that the distance measure, asymptotically in $p$, follows a Normal distribution.  
\begin{theorem}\label{theorem:distance}
  Suppose we have independent random variables, $Y_i\sim N(\mu_i,\sigma_i^2)$. Let $X=\sqrt{\sum\limits_{i=1}^pY_i^2}$, then as $p\rightarrow\infty$,
  \[
    \frac{X-\sqrt{\sum\limits_{i=1}^p(\mu_i^2+\sigma_i^2)}}{\sqrt{\frac{2\sum\limits_{i=1}^p\left(\mu_i\sigma_i\right)^2+\sum\limits_{i=1}^p\sigma_i^4+2\rho\sqrt{2\sum\limits_{i=1}^p\sum\limits_{j=1}^p\mu_i^2\sigma_i^2\sigma^4_j}}{2\sum\limits_{i=1}^p(\mu_i^2+\sigma_i^2)}}}\xrightarrow{\mathcal{D}}N(0,1)\;,
  \]
  where $\rho$ is an unknown correlation parameter (see proof).
\end{theorem}
\begin{proof} See the Supplementary Material.\end{proof}

Theorem \ref{theorem:distance} shows that, asymptotically in $p$, the distance between each time vector and a pre-specified fixed vector follows a Normal distribution. Hence, for piecewise constant time vectors, the resulting distance measure will follow a piecewise constant Normal distribution. It is common in the literature to assume that angles also follow a Normal distribution, as in \cite{Fearnhead2018}. We found by simulation, for large enough $p$, the angle measure defined in \eqref{eq:Angle} is well approximated by a Normal distribution with piecewise constant mean and variance. 

Whilst any theoretically valid univariate method could be used to detect changepoints in the mapped series, we use the PELT algorithm of \cite{Killick2012} as this is an exact and computationally efficient search. For $n\rightarrow\infty$, PELT is consistent in detecting the number and location of changes in mean and variance \citep{Tickle2019,Fisch2018}, hence, using Theorem \ref{theorem:distance}, we gain consistency of our distance measure as $p\rightarrow\infty$ also. When the Normal approximation of the distance and angle measures holds, we use the Normal likelihood as our test statistic within PELT and allow for changes in mean and variance. If $p$ is small, we may not want to make the Normal assumption. In this case, we recommend using a non-parametric test statistic, such as the empirical distribution from \cite{Zou2014} (where consistency has also been shown) as embedded within PELT in \cite{Haynes2017}. 

\subsection{GeomCP Algorithm}\label{sec:computation}
Algorithm \ref{alg:GeomCP} details the pseudo-code for GeomCP. As changepoints can manifest in both the distance and angle measure, we post-process the two sets of changepoints to obtain the final set of changes. We introduce a threshold, $\xi$, and say that a changepoint in the distance measure, $\hat{\tau}^{(d)}$, and a changepoint in the angle measure, $\hat{\tau}^{(a)}$, are deemed the same if $\left|\hat{\tau}^{(d)}-\hat{\tau}^{(a)}\right|\leq\xi$. If we determine two changepoints to be the same we set the changepoint location to be the one given by the angle measure as Section \ref{sec:SimHistograms} demonstrates, this results in more accurate changepoint locations. \textcolor{black}{The choice of $\xi$ should be set based upon the minimum distance expected between changepoints. Alternatively, $\xi$ could be set to zero and then an alternative post-processing step would be required to determine whether similar changepoint estimates correspond to the same change.}

\begin{algorithm}
  \caption{GeomCP}
  \label{alg:GeomCP}
  \begin{algorithmic}
    \Require $\bm{Y}\in\mathbb{R}^{n\times p},\;\text{threshold}=\xi,\;\textit{Univariate Cpt Method}.$
    \State \textbf{Step 1:} Centralize data by $y'_{i,j}=y_{i,j}-\left(\min\limits_i y_{i,j}-1\right)$.
    \State \textbf{Step 2:} Perform distance mapping: $\bm{y_i}\xrightarrow{\delta}d_i,\;\forall i$.
    \State \textbf{Step 3:} Perform \textit{Univariate Cpt Method} on $\boldsymbol d$ to recover cpts, $\hat{\boldsymbol\tau}^{(d)}$.
    \State \textbf{Step 4:} Perform angle mapping: $\bm{y_i}\xrightarrow{\alpha}a_i,\;\forall i$.
    \State \textbf{Step 5:} Perform \textit{Univariate Cpt Method} on $\boldsymbol a$ to recover cpts, $\hat{\boldsymbol\tau}^{(a)}$.
    \State \textbf{Step 6:} $\forall k$, \textbf{if} $\min\left|\hat{\boldsymbol\tau}^{(a)}-\hat{\tau}^{(d)}_k\right|<\xi$ \textbf{then} remove $\hat{\tau}^{(d)}_k$ from $\hat{\boldsymbol\tau}^{(d)}$.
    \Ensure $\hat{\boldsymbol\tau}=\text{sort}(\hat{\boldsymbol\tau}^{(a)},\hat{\boldsymbol\tau}^{(d)})$
  \end{algorithmic}
\end{algorithm}
One of the major downfalls of many multivariate changepoint methods is they are computationally infeasible for large $n$ and $p$. Within GeomCP, the computational cost to calculate both the distance and angle measures in \eqref{eq:Distance} and \eqref{eq:Angle} is $\mathcal{O}(np)$. \textcolor{black}{ If we implement the PELT algorithm for our univariate changepoint detection, this has expected computational cost $\mathcal{O}(n)$ under certain conditions. The main condition requires the number of changepoints to increase linearly with the number of time points, further details are given in \cite{Killick2012}. If these conditions are not satisfied, PELT has an at worst computational cost of $\mathcal{O}(n^2)$. Hence, the expected computational cost of GeomCP is $\mathcal{O}(np+n)=\mathcal{O}(np)$ (under the conditions in \cite{Killick2012}) and has at worst computational cost $\mathcal{O}\left(np+n^2)\right)=\mathcal{O}\left(n(p+n)\right)$.}

\textcolor{black}{
\subsection{Non-Normal and Dependent Data}\label{sec:assumptions}
The current problem setup assumes multivariate Normal distributed data with a diagonal covariance matrix. These assumptions are made to facilitate our theoretical analysis and result in the Normality of the mapped series. If these assumptions are broken, the geometric intuition described in Section \ref{sec:GeometricMapping} still holds, but we can say less about the theoretical properties of the mapped series.}

\textcolor{black}{Firstly, if we allow for an arbitrary covariance matrix, this describes the shape and spread of the data points. Suppose our data undergoes a change from $X_{\text{pre}}\sim N(\bm{0},\Sigma)$ to $X_{\text{post}}\sim(\bm{0},\boldsymbol\sigma\Sigma)$ this will cause the data points to spread out in the directions of the principle components. Hence, we would still expect the angles between the time vectors and the reference vector to change, revealing the change in covariance. We investigate this further in Section \ref{sec:Cov}. In fact, a Normal distributed data set with a known covariance matrix could be transformed into a Normal distributed data set with a diagonal covariance matrix (satisfying our initial problem setup) by an orthogonal transformation that aligns the axes with the principle components. Such a transformation would preserve the distances and angles by definition but requires knowledge of the true covariance structure.} 

\textcolor{black}{Alternatively, we could consider other inner products in our distance and angle mappings defined in \eqref{eq:Distance} and \eqref{eq:Angle}; here the geometric motivation of the method would remain valid. In this case, for an underlying mean change to occur without the distance measure changing, the new mean vector must remain exactly on the more general $(p-1)$-quadric in $\mathbb{R}^p$. This is still a highly non-generic requirement from a geometric prospective. In particular, we could use scalar products directly derived from the covariance matrix, such as the Mahalanobis Distance \citep{Mahalanobis1936}. In such cases, the direct relation between angles and the correlation coefficients is well known \citep{Wickens1995}. However, such inner products require an estimate of the covariance in each segment, which is non-trivial and therefore left as future work.} 

\textcolor{black}{If we allow the data to be distributed from a non-Normal distribution then we would expect changes to the first and second moment of these distributions to still manifest in the distance and angle mappings. However, being able to understand the distribution of the mapped series would be more challenging. In practice, the empirical cost function could be used within PELT \citep{Haynes2017} yet this would lead to less power in the detection of changes in the univariate series.}

\textcolor{black}{Finally, if we allowed temporal dependence between the time points this would lead to temporal dependence in the mapped series and an appropriate, cost function for PELT could be used. Understanding how the temporal dependence in the multivariate series manifests in the mapped series is non-trivial and is left as further work.
}

In the next section, we provide an extensive simulation study exploring the effectiveness of GeomCP at detecting multivariate changes in mean and variance and demonstrate an improved detection rate on current state-of-the-art multivariate changepoint methods. Furthermore, we illustrate the improved computational speed of GeomCP over current methods, especially as $n$ and $p$ grow large.

\section{Simulation Study}\label{sec:Simulation}
In this section, we provide a comparison of GeomCP; the Inspect method of \cite{Wang2018}; and the E-Divisive method of \textcolor{black}{\cite{Matteson2014}} using the statistical software R \citep{R}. First, we investigate how changes in mean and variance of time series manifest themselves in the distance and angle measures within GeomCP. We then compare GeomCP to Inspect and E-Divisive in a wide range of scenarios including dense changepoints, where the change occurs in all or a large number of dimensions, and sparse changepoints, where the change occurs in a small number of dimensions. Changes in both mean, variance and a combination of the two will be considered.

Inspect is only designed for detecting changes in mean, therefore, it will only be included in such scenarios. In addition, Inspect is designed for detecting sparse changepoints, however, Inspect `can be applied in non-sparse settings as well' \citep{Wang2018} so we also include it in the dense change in mean scenarios. Like GeomCP, E-Divisive is designed for dense changepoints, but, we will also include it in the sparse changepoint scenarios to assess performance.

For GeomCP we perform the mappings in \eqref{eq:Distance} and \eqref{eq:Angle} before applying the PELT algorithm using the \textit{changepoint} package \citep{R:PELT2014}. Unless otherwise stated, we use the default settings; namely, the MBIC penalty \citep{Zhang2007}, Normal distribution and allow for changes in mean and variance. We implement the Inspect method using the \textit{InspectChangepoint} package \citep{R:Inspect}. The thresholds used to identify significant changepoints are calculated before timing the simulations using the data-driven approach suggested in \cite{Wang2018}. For the remaining user-defined parameters, we use the default settings with $Q=0$. Setting $Q=0$ implements a Binary Segmentation approach \citep{Scott1974,Vostrikova1981} for identifying multiple changepoints. When using $Q=1000$, as suggested in \cite{Wang2018}, a Wild Binary Segmentation \citep{Fryzlewicz2014} approach is implemented to detect multiple changes. However, this becomes computational infeasible even at moderate levels of $n$ and $p$ while only resulting in minor improvements in detection rate at the expense of higher false discovery rates. For $p>1000$, the data-driven calculation of the thresholds was computationally infeasible, hence, the theoretical threshold derived in \cite{Wang2018} was originally implemented. However, this led to an excessive number of false positives and, as such, is not included. For the implementation of the E-Divisive method, we use the \textit{ecp} package \citep{R:ecp} with $\alpha=1$; minimum segment size of 30; a significance level of 0.05; and $R=499$ as suggested by \textcolor{black}{\cite{Matteson2014}}.\\

Unless indicated otherwise, we simulate data from a Normal distribution with changes in mean and variance given in each scenario. Additionally, the number of changepoints is set as $m=\left\lceil\frac{n}{200}\right\rceil$ and we distribute the changepoints \textcolor{black}{uniformly at random} throughout the time series with the condition that they are at least 30 time points apart. Where computationally feasible, we perform 500 repetitions of each scenario and display the true detection rate (TDR) and false detection rate (FDR) along with their confidence intervals given by two standard errors. For scenarios with $n\geq1000$, E-Divisive was only run on 30 replications due to the high computational cost. Changepoint estimates are deemed correct if they are the closest to, and within 10 time points of, the true changepoint and contribute to the TDR. Changepoint estimates more than 10 time points from the true changepoints or where another estimated changepoint is closer to the true changepoint are deemed false and contribute to the FDR. We seek a TDR as close to 1 as possible and an FDR as close to 0 as possible. As GeomCP estimates changepoints in both the distance and angle measures we apply the reconciling method from Section \ref{sec:computation} with the threshold, $\xi=10$. Then we apply the same TDR/FDR method to the reconciled changes.

\subsection{Size of Changepoints}\label{sec:simchangesize}
As we are interested in multivariate changepoints, we need to decide upon the size of a change in each series. If we fixed a specific change size in each series, then as $p$ increases, the change becomes easier to identify due to multivariate power. If we fixed a total change size across all series, then as $p$ increases, the change becomes considerably harder to detect. Hence, we set our simulated change sizes so that GeomCP has an approximately constant performance across $p$, in terms of TDR and FDR.

To achieve a constant performance in the change in mean scenario, we require the difference in the expected distance measure pre- and post-change to be constant across $p$. If we assume unit variance and a set mean across all series before the change, $\tilde{\mu}_{\text{pre}}$ and after the change, $\tilde{\mu}_{\text{post}}$, using Theorem \ref{theorem:distance} the expected difference in the distance measure before and after a changepoint is, 
\[
  \mathbb{E}(d_{\text{post}}-d_{\text{pre}})=\sqrt{p}\left(\sqrt{\tilde{\mu}_{\text{post}}^2+1}-\sqrt{\tilde{\mu}_{\text{pre}}^2+1}\right)\;.
\]
If we set the total mean change size in our simulated data as, 
\begin{equation}
  \sum\limits_{j=1}^p\mu_{j,\text{post}}-\mu_{j,\text{pre}}=\sqrt{p}\Theta\;,\label{eq:meanchangesize}
\end{equation}
for some constant $\Theta$ and, again, assume the mean of each series is the same, we gain,
\begin{align*}
   \Theta&=\sqrt{p}\left(\tilde{\mu}_{\text{post}}-\tilde{\mu}_{\text{pre}}\right)\\
    &\approx\sqrt{p}\left(\sqrt{\tilde{\mu}_{\text{post}}^2+1}-\sqrt{\tilde{\mu}_{\text{pre}}^2+1}\right)\\
    &=\mathbb{E}(d_{\text{post}}-d_{\text{pre}})\;.
\end{align*}
Hence, for a constant $\Theta$, using a total mean change size scaling as in \eqref{eq:meanchangesize} will result in the expected difference of the distance pre- and post-change, and therefore the performance of GeomCP, being approximately constant across $p$. As we re-scale our data before applying our two mappings, the pre- and post-change means will be large enough that this approximation is reasonable. 

Similarly, to gain an approximately constant performance of GeomCP across $p$ for a change in variance, we set the total variance change size in our simulated data as, 
\begin{equation}
  \prod\limits_{j=1}^p\frac{\sigma_{j,\text{post}}}{\sigma_{j,\text{pre}}}=\Phi^{\sqrt{p}}\;,\label{eq:varchangesize}
\end{equation}
for some constant $\Phi$. When comparing methods, we shall use \eqref{eq:meanchangesize} and \eqref{eq:varchangesize} to define the total change size for each scenario, with the change size being the same in all series that undergo a change. 

\subsection{GeomCP Investigation}\label{sec:SimHistograms}
First, we investigate how changes in mean, variance and a combination of the two, manifest themselves in the distance and angle measure within GeomCP. We set $n=1000$ and $p=200$ and simulate data with changepoints $\boldsymbol\tau=(250,500,750)$. At $\tau_1$ we have a mean change of $+0.1$ in all series; at $\tau_2$ we have a variance change of $\times1.2$ in all series; and at $\tau_3$ we have a mean change of $-0.1$ and a variance changes of $\times1.2^{-1}$ in all series. Figure \ref{fig:GeomCPExample} shows 4 of the 200 series and shows the changepoints are undetectable by eye in the individual series. Applying the mappings within GeomCP results in the mapped series seen in Figure \ref{fig:GeomCPExample} where the changes are clearly identifiable in at least one of the distance or angle measure.  

\begin{figure}[h]
  \centering
  \includegraphics[width=0.9\textwidth]{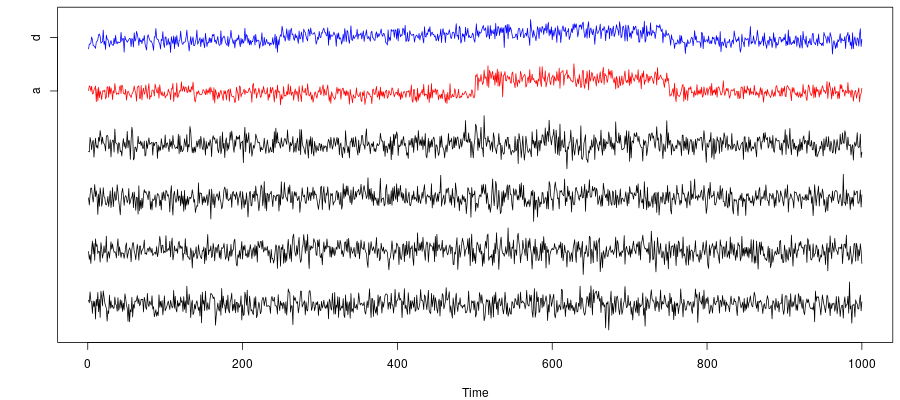}
  \caption{4 series from the simulated data set with the distance (d) and angle (a) mappings showing 3 changepoints that are not obvious in the individual series}
  \label{fig:GeomCPExample}
\end{figure}

Figure \ref{fig:DistCpts} and \ref{fig:AngCpts} shows the position of identified changepoints in the distance and angle measure in 1000 replications of the current scenario using PELT. The relatively small change in mean at time point 250 is only reliably picked up by the distance measure. The change in variance is picked up by the angle measure in almost all cases and is also seen in the distance measure, however, with less accuracy and less often. The change in mean and variance at time point 750 is reliably detected in both the distance and angle measures. These findings were similar for varying mean and variance changes. As such, this justifies setting the location of changepoints that occur in both series to be given by the angle changepoint location as stated in Section \ref{sec:computation}.

\begin{figure}[h]
  \centering
  \subfigure[Distance Changepoints]{
    \includegraphics[width=0.45\linewidth]{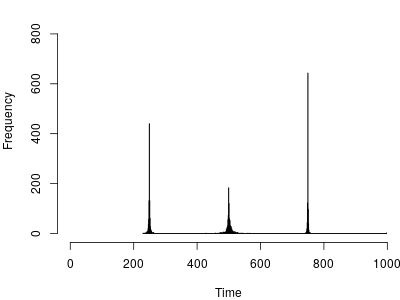}
    \label{fig:DistCpts}
  }
  \subfigure[Angle Changepoints]{
    \includegraphics[width=0.45\linewidth]{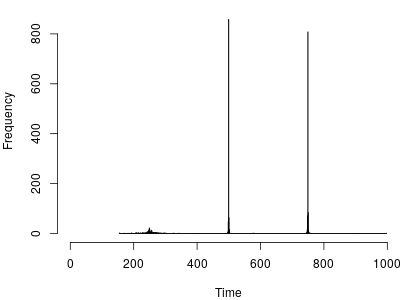}
    \label{fig:AngCpts}
  }
  \caption{Locations of detected changepoints in 1000 repetitions of simulated data set with changepoints at 250, 500 and 750, in mean, variance and, mean and variance, respectively}
\end{figure}

\subsection{Dense Changepoints}\label{sec:SimFullyMVChanges}
Now we compare GeomCP's performance with E-Divisive and Inspect. We investigate dense variance changes here, with mean, and mean and variance changes given in the Supplementary Material.

We simulate data with variance changes that occur in all series for a wide range of $n$ and $p$ and show a subset of the results here. We keep the mean vector constant and we split the total change size defined in \eqref{eq:varchangesize} evenly across all series. We display results with $\Phi=3$ as this is shown to give a high TDR while maintaining a low FDR in \cite{Eckley2011} for $p=1$. Similar findings occur with varying values of $\Phi$; \textcolor{black}{see the Supplementary material.} We apply the GeomCP and E-Divisive methods to these simulated data sets and the TDR and FDR are shown in Figure \ref{fig:VarFM}.

Figure \ref{fig:VarTDR} shows the TDR across different numbers of dimensions and time points. It is clear that GeomCP outperforms E-Divisive in terms of TDR and the gap between the methods widens as the number of dimensions increases. Figure \ref{fig:VarFDR} shows that the improved TDR of GeomCP does not come at the expense of a higher FDR, which has similar rates across $n$ and $p$. 
\begin{figure}[h]
  \centering
  \subfigure[TDR]{
    \includegraphics[width=0.45\textwidth]{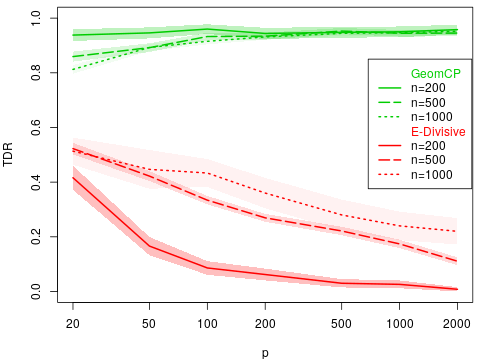}
    \label{fig:VarTDR}
  }
  \subfigure[FDR]{
    \includegraphics[width=0.45\textwidth]{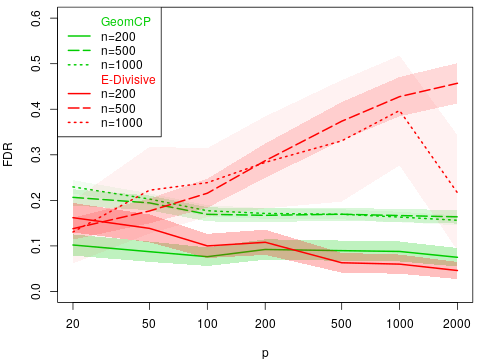}
    \label{fig:VarFDR}
  }
  \caption{TDR and FDR for GeomCP and E-Divisive for simulated data sets containing variance changes that occur in all series for multiple $n$ and $p$}
  \label{fig:VarFM}
\end{figure}

In the mean, and mean and variance change scenarios, GeomCP similarly outperforms both E-Divisive and Inspect in terms of TDR whilst maintaining a low-level FDR across $n$ and $p$. Results can be found in the Supplementary Material. 

\subsection{Sparsity Investigation}\label{sec:SimSparsity}
Thus far we assumed that all series undergo a change at each changepoint. We now explore the effect of the sparsity of the changepoint. We define $\kappa\in(0,1]$ to be the probability that a series undergoes a change. We explore sparse mean changes here, with sparse variance changes included in the Supplementary Material.

For the sparse changepoint scenarios, we set $n=500$, $p=200$ and vary $\kappa$; we note that there were similar findings for different $n$ and $p$. We keep the variance vector constant and the change size in each series that undergoes a change, is the total change size defined in \eqref{eq:meanchangesize}, split between the expected number of series to undergo a change. This means the expected total change size is the same as when all series undergo a change. We display results with $\Theta=1.2$ and similar findings occur with varying values of $\Theta$; \textcolor{black}{see the Supplementary material.} We apply the GeomCP, Inspect and E-Divisive methods to these scenarios and the TDR and FDR are shown in Figure \ref{fig:MeanSparsity}.

Figure \ref{fig:MeanSparsityTDR} shows that GeomCP maintains a constant TDR across $\kappa$ as expected. This reflects the set up of the scenario where the expected total change size is constant across $\kappa$. For dense changepoints, GeomCP compares well as we might expect. Interestingly, E-Divisive also assumes dense changepoints but performs poorly in this scenario. Inspect is designed for sparse changes and as expected, for very sparse changes the method performs the best. For sparse changepoints, the improved performance of Inspect and E-Divisive may be due to the size of change in each affected series increasing as $\kappa$ decreases.
\begin{figure}[h]
  \centering
  \subfigure[TDR]{
    \includegraphics[width=0.45\textwidth]{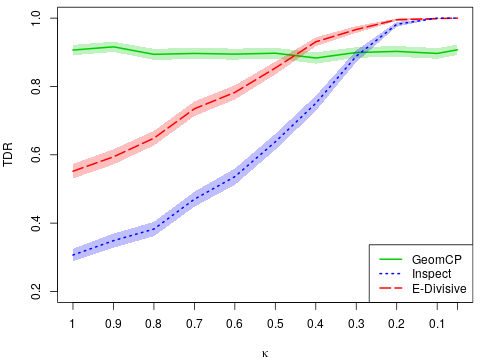}
    \label{fig:MeanSparsityTDR}
  }
  \subfigure[FDR]{
    \includegraphics[width=0.45\textwidth]{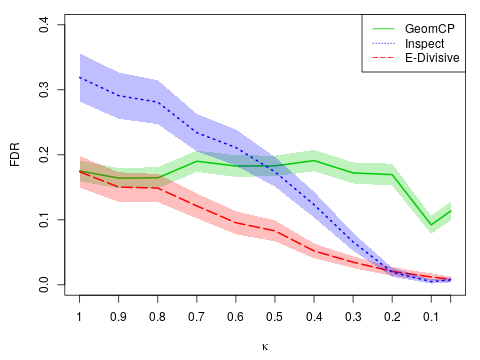}
    \label{fig:MeanSparsityFDR}
  }
  \caption{TDR and FDR for GeomCP, Inspect and E-Divisive for simulated data sets with sparse mean changes for $n=500$ and $p =200$}
  \label{fig:MeanSparsity}
\end{figure}

\textcolor{black}{
\subsection{Between-series Dependence}\label{sec:Cov}
Now we will relax the assumption of a diagonal covariance matrix and investigate how this affects the performance of GeomCP. We will investigate how two different covariance matrix structures compare to the independent, diagonal covariance case. Here we will investigate variance changes in these covariance structures with mean changes explored in the supplementary material.}

\textcolor{black}{For these scenarios, we set $n=200$, $p=100$ and have one changepoint at $\tau=100$. The pre-changepoint data will be distributed from a $N(\bm{0},\Sigma)$ while the post-changepoint data distributed from a $N(\bm{0},\boldsymbol\sigma\Sigma)$. We will vary the change size, $\boldsymbol\sigma$, while each entry of $\boldsymbol\sigma$ will be identical for each change size. We will compare three structures for $\Sigma$:
\begin{enumerate}
\item Independent case: $\Sigma=I$.
\item Block-diagonal case: Here $\Sigma$ will be a block-diagonal matrix with block size of 2. The off-diagonal entries will be randomly sampled from a $U(-0.6,-0.3)\cup U(0.3,0.6)$ distribution with the diagonal entries equal to 1.
\item Random case: Here we let $\Sigma=PDP'$ where $P$ is an orthogonalized matrix of standard Normal random variables and $D$ is a diagonal matrix with entries decreasing from 30 to 1.
\end{enumerate}
As we no longer have independence between series we cannot assume Normality of the distance and angle measures within GeomCP. Hence, we use the empirical cost function \citep{Haynes2017} within PELT to detect changes in the distance and angle measures. We similarly use the empirical cost function in the independent case for comparability.}

\textcolor{black}{Figure \ref{fig:CovTDR} shows the TDR of GeomCP and E-Divisive for varying change sizes, $\sigma$, and the different covariance structures. GeomCP clearly has a greater TDR than E-Divisive for smaller change sizes. However, Figure \ref{fig:CovFDR} shows this comes at the expense of a higher FDR. This is to be expected when using the empirical cost function within PELT as this generally produces a higher FDR. By altering the penalty used within PELT this FDR could be reduced at the cost of some power in detecting changes. Yet for $\sigma\geq 1.3$ GeomCP has a competitive FDR with E-Divisive while having a superior TDR. Interestingly, the covariance structure has very little impact on the performance of GeomCP, this follows our intuition from Section \ref{sec:assumptions}.}
\begin{figure}[h!]
  \centering
  \subfigure[TDR]{
    \includegraphics[width=0.4\textwidth]{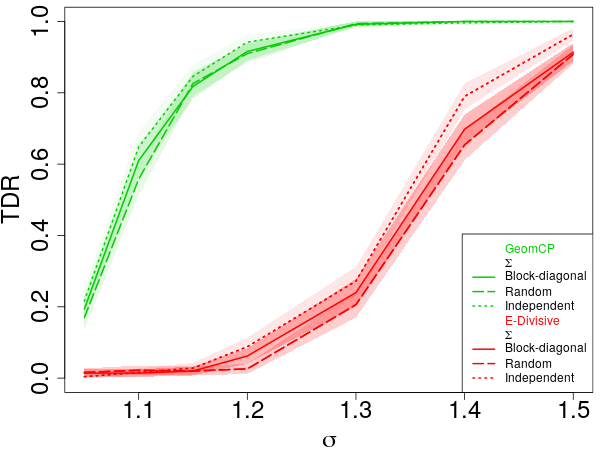}
    \label{fig:CovTDR}
  }
  \subfigure[FDR]{
    \includegraphics[width=0.4\textwidth]{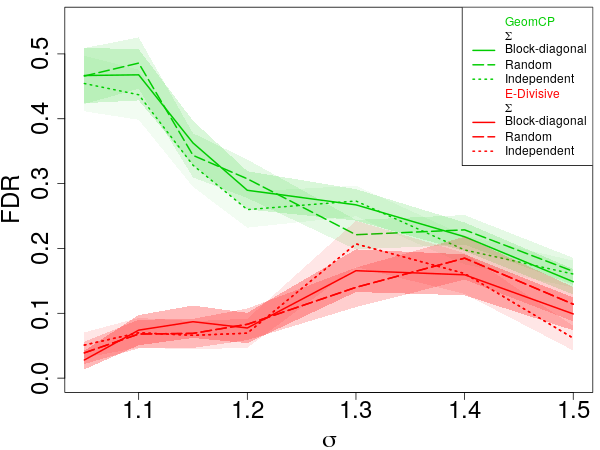}
    \label{fig:CovFDR}
  }
  \caption{TDR and FDR for GeomCP and E-Divisive for simulated data with a change in covariance for $n=200$ and $p=100$}
  \label{fig:cov}
\end{figure}

\subsection{Computational Speed}\label{sec:SimSpeed}
A major issue with high-dimensional changepoint detection is, as $n$ and $p$ grow large, many multivariate changepoint methods become computationally infeasible. Here, we compare the computational speeds of GeomCP, Inspect and E-Divisive for a range of $n$, $p$ and $m$. We will compare the speeds in three scenarios:
\begin{enumerate}
  \item $n$ increasing while $p=200$ and $m=\left\lceil\frac{n}{200}\right\rceil$.
  \item $n$ increasing while $p=200$ and $m=2$.
  \item \textcolor{black}{$p$ increasing while $n=500$ and $m=\left\lceil\frac{n}{200}\right\rceil=3$.}
\end{enumerate}
\textcolor{black}{The second scenario breaks PELT's assumption of a linearly increasing number of changepoints as the number of time points increases.} This means the speed of detecting changepoints using GeomCP will no longer be linear in time. We performed simulations using the three scenarios defined above and only included mean changes so we can compare with Inspect. We set the mean change size to be $\theta_j=0.8$ in all series so that the changes are obvious. For scenario 1 and 2, E-Divisive was computationally infeasible for $n\geq1000$. For scenario 3, Inspect's speed is only shown for $p<1000$ due to the excessive computational time of generating a data-driven threshold. Note that the data-driven thresholds needed for Inspect were calculated outside of the recorded times. In practice, if a threshold was required, then Inspect would take considerably longer to run especially as $p$ increases. Within GeomCP we run the algorithm in serial, performing the mapping and changepoint identification for the distance and then for the angle. These could be processed in parallel, leading to a further reduction in computational time. 

Figure \ref{fig:SimSpeed} shows the computational speed of each method in the three scenarios. We can see from Figure \ref{fig:SimSpeed1} that, in scenario 1, GeomCP is the fastest of the three methods for all $n$. As $n$ increases the difference between the speeds of GeomCP and Inspect increases (note the log scale on both \textcolor{black}{axes}). We can also see, E-Divisive is substantially slower than GeomCP and Inspect for all $n$ and its run time increases rapidly as $n$ gets large. Scenario 1 supports our claim that GeomCP has linear run time in $n$ when the required assumptions of PELT are met.

In scenario 2, shown in Figure \ref{fig:SimSpeed3}, we break the assumption within PELT that the number of changepoints is increasing linearly in time. This results in a comparatively slower performance of GeomCP, although, it remains computationally faster than Inspect for all $n$ shown. Similarly to scenario 1, E-Divisive has a much longer run time than both GeomCP and Inspect. 

Finally, for scenario 3 Figure \ref{fig:SimSpeed2} shows that, for small $p$, Inspect is the fastest of the methods but as $p$ increases above 50 GeomCP is computationally faster. Whilst Inspect is faster for $p<50$, recall that this does not include the time for the calculation of the threshold. Interestingly, the run time of E-Divisive appears unaffected by $p$ until $p\geq1000$. This is likely due to it's computational cost being mainly affected by the number of changepoints and time points, which remain constant. Scenario 3 also supports our claim that GeomCP has linear run time as $p$ increases, note the log scale that distorts the linearity of the plot.  

\begin{figure}[h]
  \centering
  \subfigure[]{\includegraphics[width=0.45\textwidth]{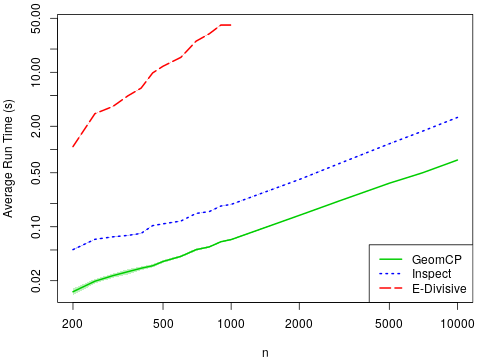}\label{fig:SimSpeed1}}
  \subfigure[]{\includegraphics[width=0.45\textwidth]{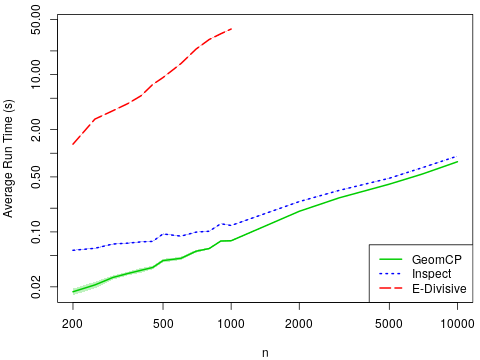}\label{fig:SimSpeed3}}
  \subfigure[]{\includegraphics[width=0.45\textwidth]{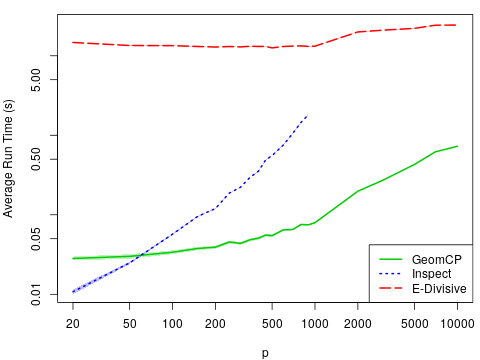}\label{fig:SimSpeed2}}
  \caption{Average run time for each method when: (a) $n$ is increasing, $p=200$ and $m$ is increasing;  (b) $n$ is increasing, $p=200$ and $m=1$; (c) $n=500$, $p$ is increasing and $m=3$ by default}
  \label{fig:SimSpeed}
\end{figure}

\section{Applications}\label{sec:Applications}
\subsection{Comparative Genomic Hybridization}\label{sec:ComparativeGenomicHybridization}
We study the comparative genomic hybridization microarray data set from \cite{Bleakley2011}. Comparative genomic hybridization allows the detection of copy number abnormalities in chromosomes by comparing the fluorescent intensity levels of DNA fragments between a test and reference sample. The data set contains log-intensity-ratio measurements from 43 individuals with bladder tumors with measurements taken at 2215 different positions on the genome. This data set is available in the \textit{ecp} R package \citep{R:ecp}.

Copy number abnormalities come in regions on the genome and can either be specific to the individual or can be shared across several individuals. It is the latter that are of more interest as these are more likely to be disease-related. E-Divisive and Inspect have both been used to segment this data set with their results shown in \textcolor{black}{\cite{Matteson2014}} and \cite{Wang2018} respectively. Under the default settings, these two methods fitted a large number of changepoints, 93 and 254 respectively, which may not be representative of changes occurring across multiple individuals. \cite{Wang2018} suggest selecting the 30 most significant changepoints to counter this, however, the justification for choosing 30 is unknown.   

To perform our analysis we first scale each series, similarly to Inspect, using the median absolute deviation to allow a better comparison. We then use the two mappings within GeomCP and apply the PELT algorithm, using the R package \textit{changepoint.np} \citep{R:EDPELT}, to the resulting mapped series. \textcolor{black}{The mappings do not appear Normal for this application, hence,} we use the empirical cost function and set the number of quantiles as $4\log(n)$, as suggested in \cite{Haynes2017}. We use the CROPS algorithm of \cite{Haynes2017b} to identify an appropriate penalty value with diagnostic plots shown in the Supplementary Material. This lead to 37 changepoints being identified and these are shown in Figure \ref{fig:GeneticsGeomCP} with the signal for 8 individuals from the study and the distance and angle mappings. Approximately $67.5\%$ of the changepoints identified by GeomCP corresponded to those identified by E-Divisive (within 3 time points), with the majority of the rest corresponding to where E-Divisive fitted two changepoints. Also, the changepoints identified seem to be common across multiple individuals while changes specific to a series are not detected. It is promising that our proposed segmentation identifies similar changepoints as other methods, while only identifying those that seem common across multiple individuals. Other GeomCP segmentations, using different potential penalty values identified in CROPS, resulted in more or less of the individual features from specific series being detected.
\begin{figure}[h]
  \centering
  \includegraphics[width=0.9\textwidth]{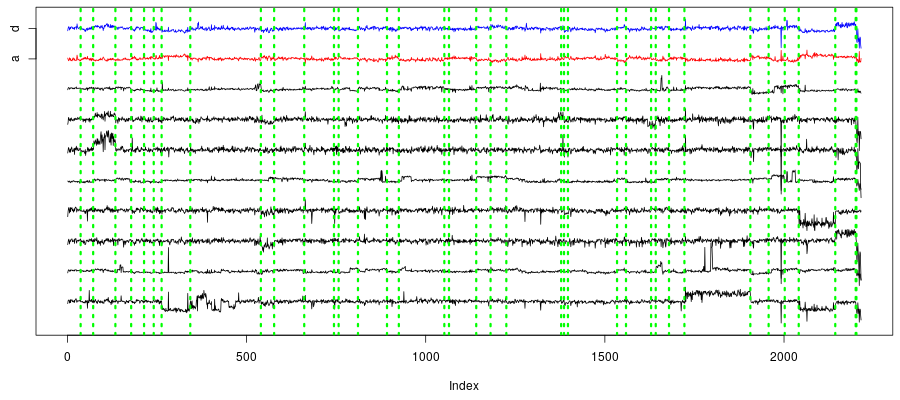}
  \caption{Log-intensity-ratio measurements of microarray data from 6 out of 43 individuals and distance (d) and angle (a) mappings with vertical lines showing the identified changepoints}
  \label{fig:GeneticsGeomCP}
\end{figure}
\subsection{S\&P500 Stock Prices}
We now investigate the daily log-returns of the closing stock prices for 447 companies included in the S\&P500 from January 2015 through to December 2016. This data set was created by \cite{Kaggle:SP500} and was loaded using the R package \textit{SP500R} \citep{R:SP500}. The aim is to identify changes in log-returns that affect a large number of companies rather than changes that are specific to individual companies. First we scale each series using the median absolute deviation. Next we apply the mappings within GeomCP, before using the PELT algorithm from the \textit{changepoint} package \citep{R:PELT2016} to both mapped series using the Normal cost function. We used the CROPS algorithm of \cite{Haynes2017b} to identify an appropriate penalty value for both series with diagnostic plots shown in the Supplementary Material. 

Using GeomCP, we identified 10 changepoints. These are shown in Figure \ref{fig:SP500} along with the log-returns of the first 10 companies from the S\&P500 list and the mapped distance and angle measures. These changepoints correspond to key events that we would expect to impact the stocks of a large number of companies. For example, the changepoints in August 2015 correspond to large falls in the Chinese stock markets with the Dow Jones industrial average falling by 1300 points over 3 days. The changepoints in February and late June 2016 likely correspond to the announcement and subsequent result of the British referendum to leave the European Union. Applying the E-Divisive method, (with the minimum segment length set to 2 and the rest of the user defined parameters set as in Section \ref{sec:Simulation}) resulted in only 2 changepoints, both occurring in August 2015 similar to those detected in GeomCP.  

\begin{figure}[h]
  \centering
  \includegraphics[width=0.9\textwidth]{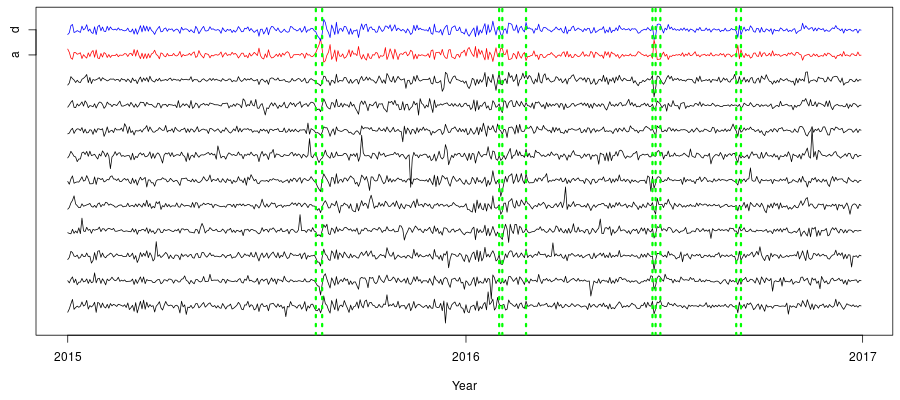}
  \caption{Log-returns of 10 out of 447 companies within the S\&P500 and the distance (d) and angle (a) mappings with vertical lines showing the identified changepoints}
  \label{fig:SP500}
\end{figure}
\section{Conclusion}\label{sec:Conclusion}
We have presented a new high-dimensional changepoint detection method that can detect mean and variance changes in multivariate time series. This is achieved by implementing a univariate changepoint detection method on two related geometric mappings of the time series. We have shown that looking at the high-dimensional changepoint problem from a geometric viewpoint allows us to utilize relevant geometric structures to detect changepoints. We have displayed an improved performance in detecting and identifying the location of multiple changepoints over current state-of-the-art methods. Furthermore, we have shown an improved computational speed over competing methods when using the univariate changepoint method PELT. Finally, we have shown the effectiveness of GeomCP at detecting changepoints when applied to applications.

\textcolor{black}{We have discussed how to extend this methodology to non-Gaussian data along with temporal and between-series dependence. However, a thorough investigation of how changes manifest in the distance and angle measure in the presence of these structures is left as future work.}

\bibliographystyle{apalike}      
\bibliography{reference}   

\end{document}


\maketitle

\section{Preliminary Lemmas}
\begin{lemma}\label{lemma:LyapunovCondition}
  Suppose we have independent random variables, $Y_i\sim\Gamma(\alpha,\beta_i)$, with a common shape parameter, $\alpha>0\;\forall i$, and varying scale parameters, $\beta_i>0\;\forall i$. Let $\mathbb{E}(Y_i)=\mu_i$, $\text{Var}(Y_i)=\sigma_i^2$ and $s_n^2=\sum\limits_{i=1}^n\sigma_i^2$. Then, $\exists\delta>0$ such that,
    \[
      \lim\limits_{n\rightarrow\infty}\frac{1}{s_n^{2+\delta}}\sum\limits_{i=1}^n\mathbb{E}\left[|Y_i-\mu_i|^{2+\delta}\right]=0\;.
    \]
\end{lemma}
\begin{proof}
  Consider the moment generating function of $Y_i-\mu_i$,
  \begin{align*}
    M_{Y_i-\mu_i}(t)&=\mathbb{E}\left[e^{t(Y_i-\mu_i)}\right]\\
    &=e^{-t\mu_i}\mathbb{E}\left[e^{tY_i}\right]\\
    &=\frac{e^{-t\mu_i}}{(1-t\beta_i)^\alpha}\;.
\end{align*}
By considering the 4th derivative of $M_{Y_i-\mu_i}(t)$ we gain,
\begin{align*}
  \sum\limits_{i=1}^n \mathbb{E}\left[|Y_i-\mu_i|^4\right]=&\sum\limits_{i=1}^nM^{(\textrm{IV})}_{Y_i-\mu_i}(0)\\
  =&\sum\limits_{i=1}^n(\alpha+3)(\alpha+2)(\alpha+1)\alpha\beta_i^4-4(\alpha+2)(\alpha+1)\alpha\beta_i^3\mu_i\\
  &+6(\alpha+1)\alpha\beta_i^2\mu_i^2-4\alpha\beta_i\mu_i^3+\mu_i^4\\\
  =&c_1\sum\limits_{i=1}^n(\beta_i^4)-c_2\sum\limits_{i=1}^n(\beta_i^3\mu_i)+c_3\sum\limits_{i=1}^n(\beta_i^2\mu_i^2)-c_4\sum\limits_{i=1}^n(\mu_i^3\beta_i)+\sum\limits_{i=1}^n\mu_i^4\\
  \leq&c_1n\beta_{\max}^4-c_2n\beta_{\min}^3\mu_{\min}+c_3n\beta_{\max}^2\mu_{\max}^2-c_4n\beta_{\min}\mu_{\min}^3+n\mu_{\max}^4\\
  =&nc_5\;,
\end{align*}
where $c_1,c_2,c_3,c_4\in\mathbb{R^+}$ and $c_5\in\mathbb{R}$.\\
Now,
\begin{align*}
  s_n^4&=\left(\sum\limits_{i=1}^n\sigma_i^2\right)^2\\
  &\geq n^2\sigma_{\min}^4\;.
\end{align*}
Hence, for $\delta=2$,
\[
  0\leq\lim\limits_{n\rightarrow\infty}\frac{1}{s_n^{2+\delta}}\sum\limits_{i=1}^n\mathbb{E}\left[|Y_i-\mu_i|^{2+\delta}\right]\leq\lim\limits_{n\rightarrow\infty}\frac{nc_5}{n^2\alpha^2\beta_{\min}^4}=0\;,
\]
as required.
\end{proof}
\begin{lemma}\label{lemma:LyapunovCLT}
  Suppose we have independent random variables, $Y_i\sim N(0,\sigma_i^2)$, and $X=\sum\limits_{i=1}^nY_i^2$. Then, as $n\rightarrow\infty$,
  \[
    \frac{X-\sum\limits_{i=1}^n\sigma_i^2}{\sqrt{2\sum\limits_{i=1}^n\sigma_i^4}}\xrightarrow{\mathcal{D}}N(0,1)\;.
  \]
\end{lemma}
\begin{proof}
  Let $Y_i=\sigma_iZ_i$, where $Z_i$ are independent standard normal random variables. Then,
  \[
    \left(\frac{Y_i}{\sigma_i}\right)^2=Z_i^2\;,
  \]
  and $Z_i^2\sim\chi^2(1)\sim\Gamma(\frac{1}{2},2)$. Hence,
  \[
    Y_i^2\sim\Gamma(\frac{1}{2},2\sigma^2_i)\;.
  \]
  Using Lemma \ref{lemma:LyapunovCondition} with $\alpha=\frac{1}{2}$, $\beta_i=2\sigma_i^2$ and $s_n^2=2\sum\limits_{i=1}^n\sigma_i^4$, we have, \[
    \lim\limits_{n\rightarrow\infty}\frac{1}{s_n^{4}}\sum\limits_{i=1}^n\mathbb{E}\left[|Y_i^2-\sigma_i^2|^{4}\right]=0\;,
  \]
  satisfying Lyapunov's condition for the Lyapunov Central Limit Theorem \citep{Lyapunov1901} for $\delta=2$. Applying Lyapunov's Central Limit Theorem gains the required result.
\end{proof}
\section{Proof of Theorem 1}
\begin{proof}
  We can re-write $X$ as,
  \[X=\sqrt{\sum\limits_{i=1}^p\mu_i^2+2\sum\limits_{i=1}^p\mu_i\sigma_iZ_i+\sum\limits_{i=1}^p\sigma_i^2Z_i^2}\;,
  \]
  where $Z_i$ are independent standard normal random variables. Let $Z^*$ and $Z^{**}$ be dependent standard normal random variables with unknown correlation $\rho$. Using Lemma \ref{lemma:LyapunovCLT} and the binomial expansion, as $p\rightarrow\infty$,
  \begin{align*}
    X\xrightarrow{\mathcal{D}}&\left(\sum\limits_{i=1}^p\left(\mu_i^2+\sigma_i^2\right)+2\sqrt{\sum\limits_{i=1}^p(\mu_i\sigma_i)^2}Z^*+\sqrt{2\sum\limits_{i=1}^p\sigma_i^4}Z^{**}\right)^{\frac{1}{2}}\\
    =&\left(\sum\limits_{i=1}^p\left(\mu_i^2+\sigma_i^2\right)\right)^{\frac{1}{2}}\left(1+\frac{2\sqrt{\sum\limits_{i=1}^p(\mu_i\sigma_i)^2}Z^*+\sqrt{2\sum\limits_{i=1}^p\sigma_i^4}Z^{**}}{\sum\limits_{i=1}^p\left(\mu_i^2+\sigma_i^2\right)}\right)^{\frac{1}{2}}\\
    =&\left(\sum\limits_{i=1}^p\left(\mu_i^2+\sigma_i^2\right)\right)^{\frac{1}{2}}\left(1+\frac{2\sqrt{\sum\limits_{i=1}^p(\mu_i\sigma_i)^2}Z^*+\sqrt{2\sum\limits_{i=1}^p\sigma_i^4}Z^{**}}{2\sum\limits_{i=1}^p\left(\mu_i^2+\sigma_i^2\right)}+\mathcal{O}\left(\frac{1}{p}\right)\right)\\
    =&\left(\sum\limits_{i=1}^p\left(\mu_i^2+\sigma_i^2\right)\right)^{\frac{1}{2}}\\
    &+\sqrt{\frac{2\sum\limits_{i=1}^p\left(\mu_i\sigma_i\right)^2+\sum\limits_{i=1}^p\sigma_i^4+2\rho\left(2\sum\limits_{i=1}^p\sum\limits_{j=1}^p\mu_i^2\sigma_i^2\sigma_j^4\right)^\frac{1}{2}}{2\sum\limits_{i=1}^p\left(\mu_i^2+\sigma_i^2\right)}}Z^{***}+\mathcal{O}\left(\frac{1}{\sqrt{p}}\right)\;,
  \end{align*}
  where $Z^{***}$ is a standard normal random variable, thus giving the required result.
\end{proof}
\section{Dense Mean Changepoints}\label{sec:SimFMMeanchange}
We simulate data with changes in mean that occur in all series for a wide range of $n$ and $p$ and show a subset of the results. We keep the variance vector constant and the total mean change size is,
\begin{equation}
  \sum\limits_{j=1}^p\mu_{j,\text{post}}-\mu_{j,\text{pre}}=\sqrt{p}\Theta\;.\label{eq:meanchangesize}
\end{equation}
We split the total change size evenly across all series and display results with $\Theta=1.2$. Similar findings occur with varying values of $\Theta$. As we assume changepoints are dense, the change size in each series violates the minimum change size assumption in Inspect; we include the results for interest. We apply the GeomCP, Inspect and E-Divisive methods to these scenarios and the true detection rate (TDR) and false detection rate (FDR) are shown in Table \ref{tab:FullyMVMean}. For $p\geq1000$, we are unable to calculate an appropriate stopping threshold for Inspect and, therefore, exclude it from such scenarios.

Table \ref{tab:FullyMVMean} clearly shows GeomCP has a greater TDR than both Inspect and E-Divisive in all scenarios. As $p$ increases, the performance of Inspect and E-Divisive drastically decreases, whereas GeomCP maintains a similar TDR. The FDRs shown in Table \ref{tab:FullyMVMean} indicate GeomCP is either better or competitive with Inspect and E-Divisive. Taking both TDR and FDR into account it is clear that GeomCP outperforms both Inspect and E-Divisive in detecting mean changes that occur in all series.
\begin{table}[]
\caption{TDR and FDR for GeomCP, Inspect and E-Divisive for simulated data sets containing mean changes that occur in all series}
\label{tab:FullyMVMean}
  \centering
\begin{tabular}{|cc|cc|cc|cc|}
\hline
\multicolumn{2}{|l|}{} & \multicolumn{2}{c|}{GeomCP} & \multicolumn{2}{c|}{Inspect} & \multicolumn{2}{c|}{E-Divisive} \\ \hline
n & p & TDR & FDR & TDR & FDR & TDR & FDR \\ \hline
\multirow{5}{*}{\rotatebox[origin=c]{90}{200}} & 50 & \textbf{0.956} & \textbf{0.086} & 0.484 & 0.163 & 0.812 & 0.095 \\
 & 100 & \textbf{0.928} & \textbf{0.101} & 0.304 & 0.150 & 0.614 & 0.159 \\
 & 500 & \textbf{0.952} & 0.080 & 0.018 & \textbf{0.026} & 0.160 & 0.118 \\
 & 1000 & \textbf{0.922} & \textbf{0.088} & - & - & 0.076 & 0.094 \\
 & 2000 & \textbf{0.952} & 0.086 & - & - & 0.046 & \textbf{0.065} \\ \hline
 \multirow{5}{*}{\rotatebox[origin=c]{90}{500}} & 50 & \textbf{0.877} & 0.189 & 0.613 & 0.191 & 0.790 & \textbf{0.119} \\
 & 100 & \textbf{0.913} & 0.168 & 0.470 & 0.247 & 0.689 & \textbf{0.138} \\
 & 500 & \textbf{0.891} & \textbf{0.189} & 0.203 & 0.462 & 0.389 & 0.211 \\
 & 1000 & \textbf{0.891} & \textbf{0.172} & - & - & 0.304 & 0.285 \\
 & 2000 & \textbf{0.896} & \textbf{0.172} & - & - & 0.230 & 0.380 \\ \hline
 \multirow{5}{*}{\rotatebox[origin=c]{90}{1000}} & 50 & \textbf{0.904} & 0.188 & 0.637 & 0.203 & 0.807 & \textbf{0.125} \\
 & 100 & \textbf{0.900} & 0.184 & 0.419 & 0.280 & 0.713 & \textbf{0.154} \\
 & 500 & \textbf{0.900} & \textbf{0.188} & 0.318 & 0.446 & 0.507 & 0.200 \\
 & 1000 & \textbf{0.906} & \textbf{0.184} & - & - & 0.367 & 0.311 \\
 & 2000 & \textbf{0.908} & \textbf{0.175} & - & - & 0.333 & 0.292 \\ \hline
 \multirow{5}{*}{\rotatebox[origin=c]{90}{2000}} & 50 & \textbf{0.880} & 0.200 & 0.634 & 0.213 & 0.747 & \textbf{0.131} \\
 & 100 & \textbf{0.894} & 0.187 & 0.561 & 0.262 & 0.730 & \textbf{0.127} \\
 & 500 & \textbf{0.892} & \textbf{0.190} & 0.343 & 0.392 & 0.463 & 0.284 \\
 & 1000 & \textbf{0.887} & \textbf{0.200} & - & - & 0.447 & 0.253 \\
 & 2000 & \textbf{0.889} & \textbf{0.197} & - & - & 0.360 & 0.281 \\ \hline
\end{tabular}
\end{table}

\section{Dense Mean and Variance Changepoints}\label{sec:FMMeanVarChange}
We simulate data with changes in mean and variance that occur in all series for multiple $n$ and $p$ and show a subset of the results. We set the total mean change size as in \eqref{eq:meanchangesize} and the total variance change size to be,
\begin{equation}
  \prod\limits_{j=1}^p\frac{\sigma_{j,\text{post}}}{\sigma_{j,\text{pre}}}=\Phi^{\sqrt{p}}\;,\label{eq:varchangesize}
\end{equation}
and split the total change sizes evenly across the series. As we have a change in both mean and variance, we would expect a smaller change in both to still be detectable. Here we set $\Theta=1$ and $\Phi=2$ and note similar finding occurred for different $\Theta$ and $\Phi$. We apply the GeomCP and E-Divisive to these data sets and the TDR and FDR are shown in Figure \ref{fig:MeanVarFM}. 

Figure \ref{fig:MeanVarTDR} shows that for all scenarios GeomCP has a better TDR than E-Divisive and the gap between the methods widens as $p$ increases. Additionally, the FDR, shown in Figure \ref{fig:MeanVarFDR}, shows that GeomCP has an improved or competitive FDR with E-Divisive across all scenarios. Considering this, it is clear that for detecting changes that occur in mean and variance simultaneously, GeomCP is to be preferred over E-Divisive.
\begin{figure}[h!]
  \centering
  \subfigure[TDR]{
    \includegraphics[width=0.45\textwidth]{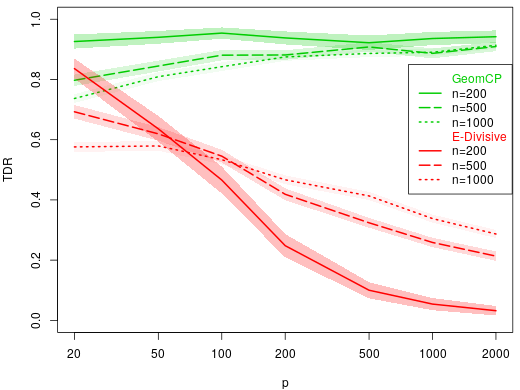}
    \label{fig:MeanVarTDR}
  }
  \subfigure[FDR]{
    \includegraphics[width=0.45\textwidth]{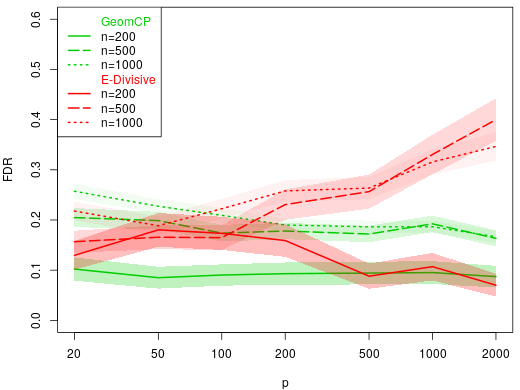}
    \label{fig:MeanVarFDR}
  }
  \caption{(a) TDR and (b) FDR for GeomCP and E-Divisive for simulated data sets containing mean and variance changes that occur in all series for multiple $n$ and $p$}
  \label{fig:MeanVarFM}
\end{figure}
\section{Sparse Variance Changepoints}\label{sec:SimSparsityVar}
To investigate sparse variance changes we set $n=500$, $p=200$ and varied $\kappa$. We keep the mean vector constant and the variance change in each series that undergoes a change, is the total variance change size defined in \eqref{eq:varchangesize}, split between the expected number number of series to undergo a change. We display results with $\Phi=3$ and note similar findings occur with varying values of $\Phi$. We apply the GeomCP and E-Divisive methods to these scenarios and the TDR and FDR are shown in Figure \ref{fig:VarSparsity}.

Figure \ref{fig:VarSparsityTDR} shows that for all levels of sparsity, GeomCP has a far greater TDR than E-Divisive. Figure \ref{fig:VarSparsityFDR} shows that GeomCP has a competitive, if not lower, FDR than E-Divisive across all sparsity levels. This shows that for sparse variance changes GeomCP has an improved performance over E-Divisive.
\begin{figure}[h!]
  \centering
  \subfigure[TDR]{
    \includegraphics[width=0.45\textwidth]{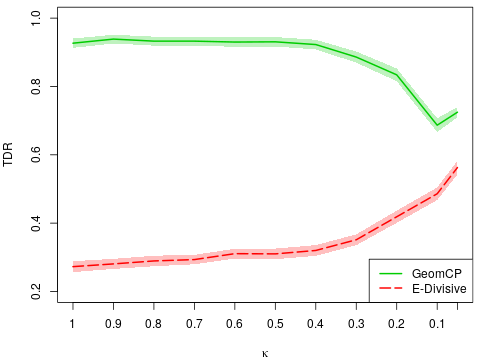}
    \label{fig:VarSparsityTDR}
  }
  \subfigure[FDR]{
    \includegraphics[width=0.45\textwidth]{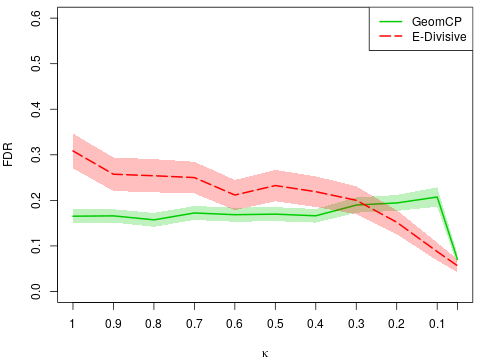}
    \label{fig:VarSparsityFDR}
  }
  \caption{(a) TDR and (b) FDR for GeomCP and E-Divisive for simulated data sets with sparse variance changes with $n=500$ and $p=200$}
  \label{fig:VarSparsity}
\end{figure}

\textcolor{black}{
  \section{Between-series Dependence: Mean Change}
  We now investigate the performance of GeomCP in scenarios where there is underlying covariance structure and a mean change occurs. For these scenarios, we set $n=200$, $p=100$ and have one changepoint at $\tau=100$. The pre-changepoint data will be distributed from a $N(\bm{0},\Sigma)$ while the post-changepoint data distributed from a $N(\boldsymbol\mu,\Sigma)$. We will vary the change size, $\boldsymbol\mu$, while the entries of $\boldsymbol\mu$ will be identical for each change size. We will compare three structures for $\Sigma$:
\begin{enumerate}
\item Independent case: $\Sigma=I$.
\item Block-diagonal case: Here $\Sigma$ will be a block-diagonal matrix with block size of 2. The off-diagonal entries will be randomly sampled from a $U(-0.6,-0.3)\cup U(0.3,0.6)$ distribution with the diagonal entries equal to 1.
\item Random case: Here we let $\Sigma=PDP'$ where $P$ is an orthogonalized matrix of standard Normal random variables and $D$ is a diagonal matrix with entries decreasing from 30 to 1.
\end{enumerate}
As we no longer have independence between series we cannot assume Normality of the distance and angle measures within GeomCP. Hence, we use the empirical cost function \citep{Haynes2017} within PELT to detect changes in the distance and angles measures. We similarly use the empirical cost function in the independent case for comparability.}

\textcolor{black}{Figure \ref{fig:CovMeanTDR} shows GeomCP has a superior TDR over E-divisive for smaller change sizes $\mu$. Interestingly, the TDR for the random covariance structure is poor for both methods. Similarly, to the case of a variance change, Figure \ref{fig:CovMeanFDR} shows by using the empirical cost function within PELT we get a worse FDR for smaller change sizes. However, this trade-off between TDR and FDR could be improved by tuning the penalty used within PELT.}

\begin{figure}[h!]
  \centering
  \subfigure[TDR]{
    \includegraphics[width=0.45\textwidth]{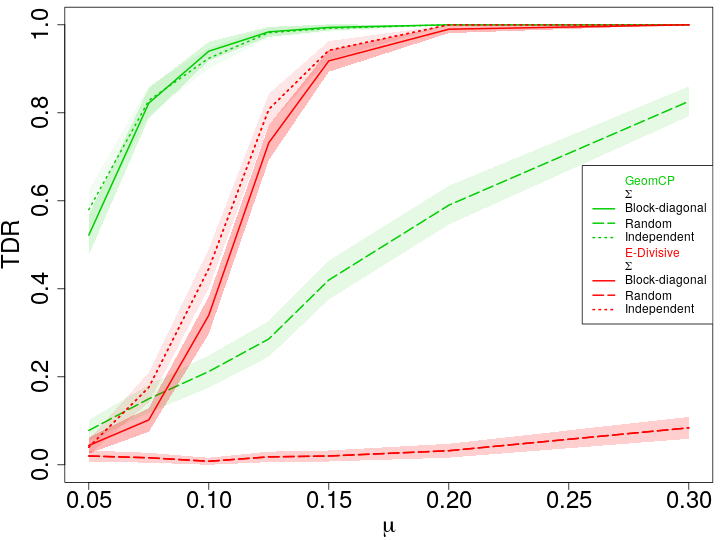}
    \label{fig:CovMeanTDR}
  }
  \subfigure[FDR]{
    \includegraphics[width=0.45\textwidth]{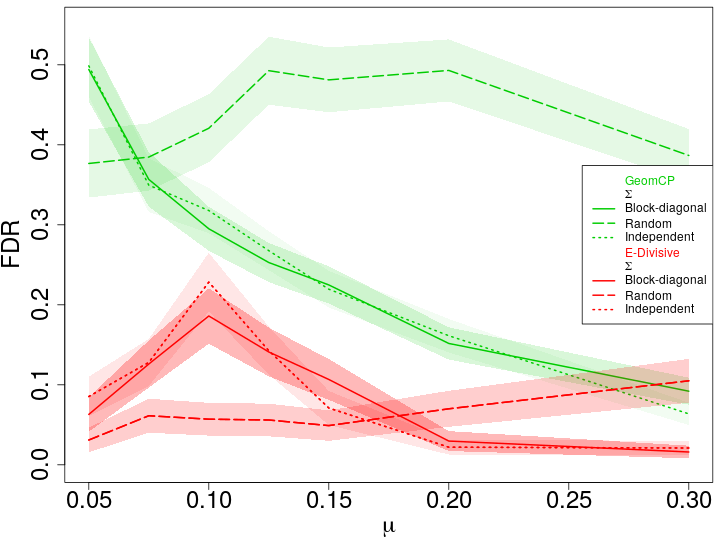}
    \label{fig:CovMeanFDR}
  }
  \caption{TDR and FDR for GeomCP (using the empirical cost function) and E-Divisive for simulated data with an underlying covariance structure and a change in mean for $n=200$ and $p=100$}
  \label{fig:covMean}
\end{figure}

\textcolor{black}{
\section{Performance under the Null}}
\textcolor{black}{Here we investigate the performance of GeomCP on null data. The aim is to keep the number of false positive as low as possible. We simulated Normal data with no changepoints for varying $n$ and $p$ and ran the GeomCP method using the Normal cost function within PELT. We calculated the false positive rate (FPR) by taking the total number of detected changepoints across all replications and dividing by the total number of replications.}

\textcolor{black}{Figure \ref{fig:Null} shows that for the majority of scenarios the FPR stays below 0.05 indicating a conservative performance in terms of the type 1 error.
\begin{figure}[h!]
  \centering
  \includegraphics[width=0.9\textwidth]{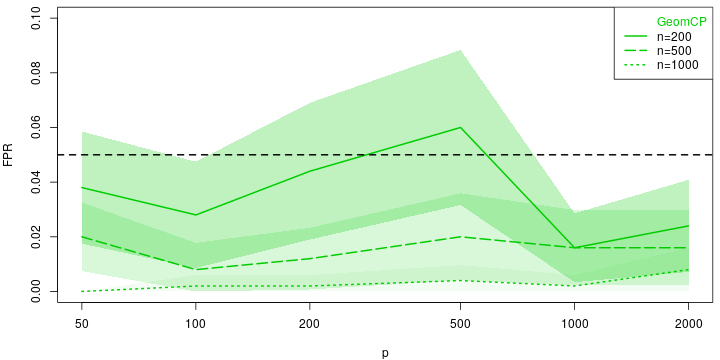}
  \caption{FPR for GeomCP (using the Normal cost function) for simulated data sets with no changepoints for varying $n$ and $p$}
  \label{fig:Null}
\end{figure}}
\textcolor{black}{
\section{$\Theta$ and $\Phi$ Investigation}
Now we investigate how the performance of GeomCP, and the competing methods, vary as we alter $\Theta$ and $\Phi$ for the dense change in mean and change in variance scenarios. Table \ref{tab:ThetaChange} shows the TDR and FDR of GeomCP, Inspect and E-Divisive for multiple $n$ and $p$ for different change size values. GeomCP's performance remains an improvement upon the competing methods for all change sizes. Table \ref{tab:PhiChange} shows similar finding for the change in variance scenarios.}
\begin{table}[h!]
  \caption{TDR and FDR for GeomCP, Inspect and E-Divisive for simulated data sets containing mean changes that occur in all series. $\Theta$ relates to the size of the change.}
  \label{tab:ThetaChange}
  \centering
\begin{tabular}{|ccc|cc|cc|cc|}
\hline
\multicolumn{3}{|l|}{} & \multicolumn{2}{c|}{GeomCP} & \multicolumn{2}{c|}{Inspect} & \multicolumn{2}{c|}{E-Divisive} \\ \hline
n & p & $\Theta$ & TDR & FDR & TDR & FDR & TDR & FDR \\ \hline
\multirow{6}{*}{\rotatebox[origin=c]{90}{200}} & \multirow{3}{*}{\rotatebox[origin=c]{90}{100}} & 1 & \textbf{0.794} & \textbf{0.092} & 0.130 & \textbf{0.092} & 0.384 & 0.181 \\
 &  & 1.2 & \textbf{0.928} & \textbf{0.101} & 0.304 & 0.150 & 0.614 & 0.159 \\
 &  & 1.5 & \textbf{0.988} & \textbf{0.073} & 0.610 & 0.180 & 0.908 & 0.076 \\ \cline{2-9} 
 & \multirow{3}{*}{\rotatebox[origin=c]{90}{500}} & 1 & \textbf{0.758} & 0.144 & 0.010 & \textbf{0.014} & 0.078 & 0.111 \\
 &  & 1.2 & \textbf{0.952} & 0.080 & 0.018 & \textbf{0.026} & 0.160 & 0.118 \\
 &  & 1.5 & \textbf{0.986} & 0.072 & 0.088 & \textbf{0.049} & 0.372 & 0.153 \\ \hline
 \multirow{6}{*}{\rotatebox[origin=c]{90}{1000}} & \multirow{3}{*}{\rotatebox[origin=c]{90}{100}} & 1 & \textbf{0.764} & 0.252 & 0.329 & 0.362 & 0.567 & \textbf{0.184} \\
 &  & 1.2 & \textbf{0.900} & 0.184 & 0.419 & 0.280 & 0.713 & \textbf{0.154} \\
 &  & 1.5 & \textbf{0.983} & 0.146 & 0.581 & 0.167 & 0.820 & \textbf{0.092} \\ \cline{2-9} 
 & \multirow{3}{*}{\rotatebox[origin=c]{90}{500}} & 1 & \textbf{0.751} & \textbf{0.263} & 0.228 & 0.538 & 0.327 & 0.369 \\
 &  & 1.2 & \textbf{0.900} & \textbf{0.188} & 0.318 & 0.446 & 0.507 & 0.200 \\
 &  & 1.5 & \textbf{0.976} & 0.144 & 0.393 & 0.412 & 0.667 & \textbf{0.136} \\ \hline
\end{tabular}
\end{table}

\begin{table}[h!]
  \caption{TDR and FDR for GeomCP and E-Divisive for simulated data sets containing variance changes that occur in all series. $\Phi$ relates to the size of the change.}
  \label{tab:PhiChange}
  \centering
\begin{tabular}{|ccc|cc|cc|}
\hline
\multicolumn{3}{|l|}{} & \multicolumn{2}{c|}{GeomCP} & \multicolumn{2}{c|}{E-Divisive} \\ \hline
n & p & $\Phi$ & TDR & FDR & TDR & FDR \\ \hline
\multirow{6}{*}{\rotatebox[origin=c]{90}{200}} & \multirow{3}{*}{\rotatebox[origin=c]{90}{100}} & 2.5 & \textbf{0.844} & 0.113 & 0.054 & \textbf{0.064} \\
 &  & 3 & \textbf{0.960} & \textbf{0.076} & 0.086 & 0.100 \\
 &  & 3.5 & \textbf{0.966} & \textbf{0.072} & 0.188 & 0.148 \\ \cline{2-7} 
 & \multirow{3}{*}{\rotatebox[origin=c]{90}{500}} & 2.5 & \textbf{0.882} & 0.100 & 0.016 & \textbf{0.070} \\
 &  & 3 & \textbf{0.948} & 0.090 & 0.030 & \textbf{0.063} \\
 &  & 3.5 & \textbf{0.976} & \textbf{0.061} & 0.044 & 0.079 \\ \hline
 \multirow{6}{*}{\rotatebox[origin=c]{90}{1000}} & \multirow{3}{*}{\rotatebox[origin=c]{90}{100}} & 2.5 & \textbf{0.804} & 0.225 & 0.373 & \textbf{0.247} \\
 &  & 3 & \textbf{0.915} & \textbf{0.177} & 0.433 & 0.239 \\
 &  & 3.5 & \textbf{0.962} & 0.166 & 0.500 & \textbf{0.156} \\ \cline{2-7} 
 & \multirow{3}{*}{\rotatebox[origin=c]{90}{500}} & 2.5 & \textbf{0.844} & \textbf{0.204} & 0.173 & 0.506 \\
 &  & 3 & \textbf{0.944} & \textbf{0.169} & 0.280 & 0.331 \\
 &  & 3.5 & \textbf{0.982} & \textbf{0.157} & 0.367 & 0.267 \\ \hline
\end{tabular}
\end{table}

\section{CROPS diagnostics plots}
\textcolor{black}{Here we show the CROPS diagnostic plots used to find the optimal number of changepoints in both our applications, the CGH data set and the S\&P500 data set. The circled points indicates the elbow of the plot and the number of changepoints we used.}
\begin{figure}[h!]
  \centering
  \subfigure[Distance]{
    \includegraphics[width=0.45\textwidth]{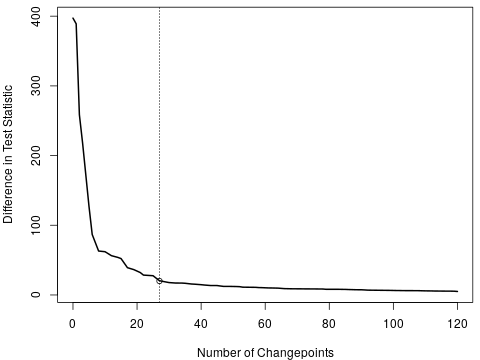}
    \label{fig:GeneticsCROPSDist}
  }
  \subfigure[Angle]{
    \includegraphics[width=0.45\textwidth]{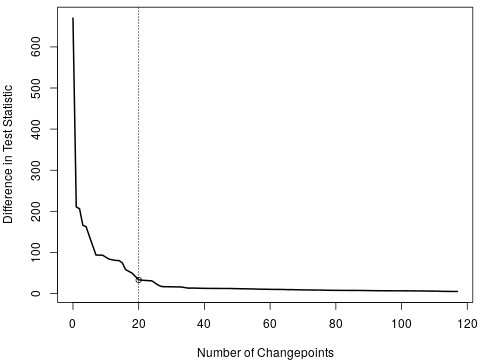}
    \label{fig:GeneticsCROPSAng}
  }
  \caption{CROPS diagnostics plots for distance and angle measure of comparative genomic hybridization data where the vertical line and circled point indicates the elbow of the plot we use for the number of changepoints}
  \label{fig:GeneticsCROPS}
\end{figure}
\begin{figure}[h]
  \centering
  \subfigure[Distance]{
    \includegraphics[width=0.45\textwidth]{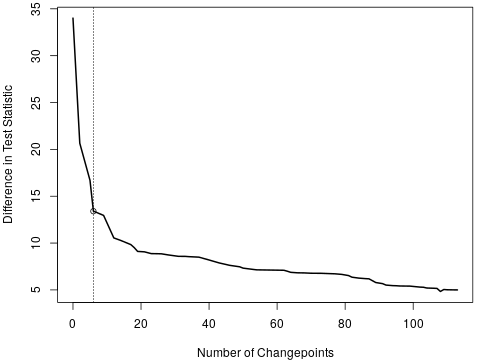}
    \label{fig:SP500CROPSDist}
  }
  \subfigure[Angle]{
    \includegraphics[width=0.45\textwidth]{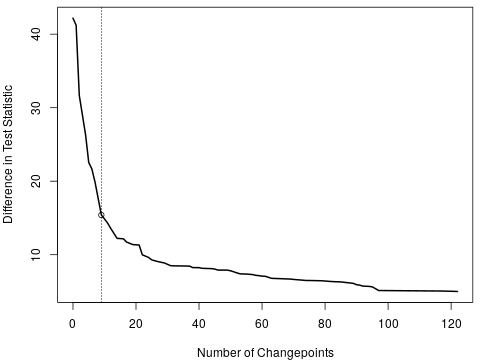}
    \label{fig:SP500CROPSAng}
  }
  \caption{CROPS diagnostics plots for distance and angle measure of S\&P500 log-returns data where the vertical line and circled point indicates the elbow of the plot we use for the number of changepoints}
  \label{fig:SP500CROPS}
\end{figure}
\bibliographystyle{apalike}
\bibliography{reference.bib}